\documentclass[nonacm,sigconf]{acmart}

\usepackage{algorithm}
\usepackage[noend]{algpseudocode}
\usepackage{amsmath}
\usepackage{amssymb}
\usepackage{amsthm}
\usepackage{enumitem}
\usepackage{tabularx}
\usepackage{xspace}
\usepackage{tikz}
\usepackage{tikz-qtree}
\usetikzlibrary{arrows,arrows.meta,calc,ipe,patterns,positioning,shapes.geometric,trees}

\definecolor{red}{rgb}{1,0,0}
\definecolor{green}{rgb}{0,1,0}
\definecolor{blue}{rgb}{0,0,1}
\definecolor{yellow}{rgb}{1,1,0}
\definecolor{orange}{rgb}{1,0.647,0}
\definecolor{gold}{rgb}{1,0.843,0}
\definecolor{purple}{rgb}{0.627,0.125,0.941}
\definecolor{gray}{rgb}{0.745,0.745,0.745}
\definecolor{brown}{rgb}{0.647,0.165,0.165}
\definecolor{navy}{rgb}{0,0,0.502}
\definecolor{pink}{rgb}{1,0.753,0.796}
\definecolor{seagreen}{rgb}{0.18,0.545,0.341}
\definecolor{turquoise}{rgb}{0.251,0.878,0.816}
\definecolor{violet}{rgb}{0.933,0.51,0.933}
\definecolor{darkblue}{rgb}{0,0,0.545}
\definecolor{darkcyan}{rgb}{0,0.545,0.545}
\definecolor{darkgray}{rgb}{0.663,0.663,0.663}
\definecolor{darkgreen}{rgb}{0,0.392,0}
\definecolor{darkmagenta}{rgb}{0.545,0,0.545}
\definecolor{darkorange}{rgb}{1,0.549,0}
\definecolor{darkred}{rgb}{0.545,0,0}
\definecolor{lightblue}{rgb}{0.678,0.847,0.902}
\definecolor{lightcyan}{rgb}{0.878,1,1}
\definecolor{lightgray}{rgb}{0.827,0.827,0.827}
\definecolor{lightgreen}{rgb}{0.565,0.933,0.565}
\definecolor{lightyellow}{rgb}{1,1,0.878}
\definecolor{black}{rgb}{0,0,0}
\definecolor{white}{rgb}{1,1,1}

\newlength\figureheight
\newlength\figurewidth

\tikzstyle{ipe stylesheet} = [
  ipe import,
  even odd rule,
  line join=round,
  line cap=butt,
  ipe pen normal/.style={line width=0.4},
  ipe pen heavier/.style={line width=0.8},
  ipe pen fat/.style={line width=1.2},
  ipe pen ultrafat/.style={line width=2},
  ipe pen normal,
  ipe mark normal/.style={ipe mark scale=3},
  ipe mark large/.style={ipe mark scale=5},
  ipe mark small/.style={ipe mark scale=2},
  ipe mark tiny/.style={ipe mark scale=1.1},
  ipe mark normal,
  /pgf/arrow keys/.cd,
  ipe arrow normal/.style={scale=7},
  ipe arrow large/.style={scale=10},
  ipe arrow small/.style={scale=5},
  ipe arrow tiny/.style={scale=3},
  ipe arrow normal,
  /tikz/.cd,
  ipe arrows,
  <->/.tip = ipe normal,
  ipe dash normal/.style={dash pattern=},
  ipe dash dashed/.style={dash pattern=on 4bp off 4bp},
  ipe dash dotted/.style={dash pattern=on 1bp off 3bp},
  ipe dash dash dotted/.style={dash pattern=on 4bp off 2bp on 1bp off 2bp},
  ipe dash dash dot dotted/.style={dash pattern=on 4bp off 2bp on 1bp off 2bp on 1bp off 2bp},
  ipe dash normal,
  ipe node/.append style={font=\normalsize},
  ipe stretch normal/.style={ipe node stretch=1},
  ipe stretch TikZ-normal/.style={ipe node stretch=1},
  ipe stretch normal,
  ipe opacity 10/.style={fill opacity=0.1},
  ipe opacity 30/.style={fill opacity=0.3},
  ipe opacity 50/.style={fill opacity=0.5},
  ipe opacity 75/.style={fill opacity=0.75},
  ipe opacity opaque/.style={opacity=1},
  ipe opacity opaque,
]

\newcommand{\lftj}{Leapfrog-Triejoin\xspace}

\newcommand{\rtj}{RTJ\xspace}
\newcommand{\radix}{Radix-Triejoin\xspace}

\DeclareMathOperator*{\bigJoin}{\Join}

\newcommand{\cR}{\mathcal{R}}
\newcommand{\cA}{\mathcal{A}}
\newcommand{\cB}{\mathcal{B}}
\newcommand{\cH}{\mathcal{H}}
\newcommand{\cE}{\mathcal{E}}
\newcommand{\cS}{\mathcal{S}}
\newcommand{\cV}{\mathcal{V}}

\newcommand{\bD}{\mathbf{D}}
\newcommand{\bU}{\mathbf{U}}

\allowdisplaybreaks

\AtBeginDocument{%
  \providecommand\BibTeX{{%
    \normalfont B\kern-0.5em{\scshape i\kern-0.25em b}\kern-0.8em\TeX}}}

\setcopyright{none}

\renewcommand\footnotetextcopyrightpermission[1]{} 
\begin{document}

\title{Worst-Case Optimal Radix Triejoin}
\author{Alan Fekete}
\affiliation{%
  \institution{The University of Sydney}
  \city{Sydney}
  \country{Australia}
}

\author{Brody Franks}
\affiliation{%
  \institution{The University of Sydney}
  \city{Sydney}
  \country{Australia}
}

\author{Herbert Jordan}
\affiliation{%
  \institution{Universit\"at Innsbruck}
  \city{Innsbruck}
  \country{Austria}
}

\author{Bernhard Scholz}
\affiliation{%
  \institution{The University of Sydney}
  \city{Sydney}
  \country{Australia}
}

\renewcommand{\shortauthors}{A.~Fekete, B.~Franks, H.~Jordan, and B.~Scholz}

\begin{abstract}
Relatively recently, the field of join processing has been swayed by the discovery of a new class of multi-way join algorithms. The new algorithms join multiple relations simultaneously rather than perform a series of pairwise joins. The new join algorithms satisfy stronger worst-case runtime complexity guarantees than any of the existing approaches based on pairwise joins -- they are worst-case optimal in data complexity. These research efforts have resulted in a flurry of papers documenting theoretical and some practical contributions. However, there is still the quest of making the new worst-case optimal join algorithms truly practical in terms of (1) ease of implementation and (2) secondary index efficiency in terms of number of indexes created to answer a query.

In this paper, we present a simple worst-case optimal multi-way join algorithm called the radix triejoin. 
Radix triejoin uses a binary encoding for reducing the domain of a database. 
Our main technical contribution is that domain reduction allows a bit-interleaving of attribute values that gives 
rise to a query-independent relation representation, permitting the computation of multiple queries over the same 
relations worst-case optimally without having to construct additional secondary indexes.
We also generalise the core algorithm to conjunctive queries with inequality constraints and provide a new proof technique for the worst-case optimal join result.
\end{abstract}

\maketitle

\section{Introduction} \label{chap:introduction}

Join processing is one of the most studied problems in computer science.
Joins are at the heart of relational database queries, and are also
known to be applicable in various fields of computer science including
problems in graph theory~\cite{nguyen2015join,aberger2017emptyheaded},  large-scale data analytics \citep{ngo2012worst,ngo2014beyond},
social network analysis \citep{kara2019counting}, inference \citep{abo2016faq}, constraint satisfaction \citep{kolaitis2000conjunctive}, coding theory \citep{gilbert2013ell}, and machine learning \citep{schleich2016learning}.
Traditionally, multi-way joins have been evaluated by a query plan composed of pairwise joins. However, it is known that the pairwise plans are asymptotically suboptimal in worst-case runtime complexity, and
in the last decade a new algorithm class has emerged --
the algorithm class of \emph{worst-case optimal join} algorithms.

The fundamental breakthrough that precipitated worst-case optimal joins was the work of 
\citeauthor*{atserias2008size}  establishing what is now known as the AGM bound: a tight worst-case output size bound
for a given join query in terms of its input relation sizes and the query's structural properties ~\cite{atserias2008size}.
Motivated by the AGM bound, a new class of join algorithms has been devised, which are ``worst-case optimal''; these exhibit, for any given query, a runtime whose worst case over database instances of a given size, coincides with the
AGM bound on the maximum output size of that join query (hiding constants and single query-expression-size factors and a logarithmic factor in the data size).
Examples of these algorithms are NPRR ~\cite{ngo2012worst} and leapfrog triejoin \citep{veldhuizen2014triejoin}.
There have been some practical studies of worst-case optimal join algorithms
\citep{nguyen2015join,aberger2017emptyheaded,chu2015theory}. 

Current worst-case optimal join algorithms suffer a drawback in that they often require the database to have a large number of secondary indexes; this  places demands on computation and memory.
Our work addresses this concern by offering a new approach of performing worst-case optimal joins. Our new
algorithm departs from the comparison-based paradigm of the previous algorithms.
Our algorithm has an analogy to radix sort,
as it uses properties of the key values themselves rather than comparisons.

To illustrate our new \emph{Worst-Case Optimal Radix Triejoin}, consider the triangle query
$Q(A, B, C)$ $:=$ $R(A, B)$  $\Join$ $S(B, C)$ $\Join$ $T(A, C)\text{.}$
For this query, traditional pairwise join evaluation plans are suboptimal with a 
worst-case execution time of $\Omega(N^2)$ 
for relation sizes $|R|, |S|, |T| \leq N$. 
It is known via the AGM bound that for the triangle query a tight bound on the maximum number of triangles
is $O(N^{1.5})$ and hence a worst-case optimal join algorithm exhibits a worst-case execution time
of $O(N^{1.5} \log N)$ for the triangle query (cf.~\citep{ngo2013skew}). In general, the AGM bound can be written as $O(N^{\rho^*(Q)})$ where $N$ is the size of the relations and $\rho^*(Q)$ is the 
query complexity obtained by a linear program.

At a high-level view, radix triejoin is based on two stages:
The first stage performs \emph{Booleanisation}~(cf.~\cite{atserias2007conjunctive}) of relations where for a relation $R$ we define an equivalent
relation $R^{(w)}$ with $w$ times more attributes, each of which is taken from the Boolean domain ${\mathbb B}=\{0, 1\}$.
A new query is defined on the Booleanised relations.
The second stage solves the Booleanised query by using the generic backtracking, divide-and-conquer approach from prior work
\cite{ngo2013skew}.

The Booleanisation of a query conceptually underpins the representation of each value in the database as a bitstring, by a suitable 
encoding of $w$ bits. For each attribute in a database schema,
we create $w$ new attributes that are associated with positions in the bitstring.
The values of the new attributes are reduced to simply $0-1$ values.
For example, with encoding length $w = 2$, the Booleanisation of the triangular query  is
$Q^{(2)}(A_0, A_1, B_0, B_1, C_0, C_1)$ 
$:=$  
$R^{(2)}(A_0, A_1, B_0, B_1)$
$\Join$ $S^{(2)}(B_0, B_1, C_0, C_1)$
$\Join$ $T^{(2)}(A_0, A_1, C_0, C_1)\text{.}$
An example database for this query on domain $\{0, 1, 2\}$ is
$R(A, B)$ $=$ $\{(0, 1)\}$, $S(B, C)$ $=$ $\{(1, 2)\}$, $T(A, C)$ $=$ $\{(0, 2)\}$
and a possible Booleanisation of the relations could be
$R^{(2)}(A_0, A_1, B_0, B_1)$ $=$ $\{(0, 0, 0, 1)\}$, 
$S^{(2)}(B_0, B_1, C_0, C_1)$ $=$ $\{(0, 1, 1, 0)\}$,
$T^{(2)}(A_0, A_1, C_0, C_1)$ $=$ $\{(0, 0, 1, 0)\}$,
where we have encoded the original domain $\{0, 1, 2\}$ of the database in the binary numeral system.

The second stage of radix triejoin is to solve the Booleanisation of the query
using a divide-and-conquer, backtracking approach. At a high-level, we
follow the recursive query decomposition approach that is known as the
\emph{generic framework} \citep{ngo2013skew}. 
Compared to traditional pairwise join algorithms, which employ a \emph{relation-based} approach,
the generic framework utilises an \emph{attribute-based} search:
satisfying assignments for the query are found by initially starting with no information
about the answer, and at a given step a candidate solution is extended by a possible binding for one of the remaining attributes.
Continuing the triangle query example, we first select an attribute arbitrarily, say $A_0$, then compute the solutions
$a_0$ of the subquery $L_1 := \pi_{A_0}(R^{(2)}) \Join \pi_{A_0}(T^{(2)})$.
The fundamental property of the subquery is that if $a_0$ is part of an answer of the full query, then $a_0$ is in the
answer of the subquery.
Since the subquery $L_1$ is over only one attribute,
we just check both possible output values ($0$ or $1$) to compute the subquery answer.
In the second step, we select another attribute, e.g., $B_0$, and for each $a_0$ we have computed,
compute the set of $b_0 \in L_2 := \pi_{B_0} (\sigma_{A_0 = a_0} (R^{(2)})) \Join \pi_{B_0} (S^{(2)})$.
Again, the subquery $L_2$ involves only one non-bound attribute making it simple to compute. The process
continues for the total of six Boolean attributes in this example. Algorithm~\ref{alg:triangle-join}
illustrates the full process as nested for-loops.
\algrenewcommand\algorithmicindent{1.1mm}%
\begin{algorithm}[t]
\caption{Booleanized Triangular Join Query}
\begin{algorithmic}[1]
\footnotesize
\Require $R^{(2)}(A_0, A_1, B_0, B_1), S^{(2)}(B_0, B_1, C_0, C_1), T^{(2)}(A_0, A_1, C_0, C_1)$
\State $Q^{(2)} \gets  \emptyset$
\State $L_1 \gets \pi_{A_0}(R^{(2)}) \Join \pi_{A_0}(T^{(2)})$
\For{every $a_0 \in L_1$}
\State $L_2 \gets \pi_{B_0} (\sigma_{A_0 = a_0} (R^{(2)})) \Join \pi_{B_0} (S^{(2)})$
	\For{every $b_0 \in L_2$}
\State $L_3 \gets \pi_{C_0} (\sigma_{B_0 = b_0} (S^{(2)})) \Join \pi_{C_0} (\sigma_{A_0 = a_0} (T^{(2)}))$
		\For{every $c_0 \in L_3$}
\State $L_4 \gets \pi_{A_1}(\sigma_{(A_0, B_0) = (a_0,  b_0)} (R^{(2)})) \Join \pi_{A_1}(\sigma_{(A_0,C_0) = (a_0, c_0)} (T^{(2)}))$
			\For{every $a_1 \in L_4$}
\State $L_5 \gets \pi_{B_1} (\sigma_{(A_0,A_1,B_0) = (a_0, a_1, b_0)} (R^{(2)}))  \Join \pi_{B_1} (\sigma_{(B_0,C_0) = (b_0,c_0)} (S^{(2)}))$
				\For{every $b_1 \in L_5$}
\State  $L_6 \gets \pi_{C_1} \left ( \begin{array}{l} \sigma_{(B_0,B_1,C_0) = (b_0,b_1,c_0)} (S^{(2)})) \Join \\ \pi_{C_1} (\sigma_{(A_0, A_1, C_0) = (a_0,a_1,c_0)} (T^{(2)}) \end{array} \right)$
					\For{every $c_1 \in L_6$}
			\State  $Q^{(2)} \gets Q^{(2)} \cup \{(a_0, b_0, c_0, a_1, b_1, c_1)\}$.
		\EndFor
		\EndFor
	\EndFor
		\EndFor
	\EndFor
\EndFor
\end{algorithmic}
\label{alg:triangle-join}
\end{algorithm}

In this work, we introduce a generic Radix Triejoin algorithm called gRTJ, which permits arbitrary attribute orders. 
We formally prove the correctness of gRTJ using a subquery recurrence for reducing the search space.
We show that gRTJ on Booleanised queries is ``worst-case optimal''.
It achieves a worst-case execution runtime coinciding with the AGM bound (up
to a factor of the encoding length and query expression size, 
essentially equivalent to prior worst-case optimal algorithms). 
We present a new analysis of runtime for the generic framework, which proves an instance bound 
on the runtime of gRTJ (as well as e.g. leapfrog triejoin, which also instantiates 
the generic framework). 
Our bound is an exact instance bound. As far as we know, this is the first
such result. Our presentation is data-structure independent, but we suggest that
bitwise tries, quadtrees, or ordered binary decision diagrams (OBDDs) are suitable index structures for the relations.

Later in this work, we extend gRTJ to conjunctive queries with inequalities and query-independent representations. We call the extension the RTJ algorithm.  
A technical decision of attribute-based multi-way join algorithms so far is that the attributes of a query are 
processed in a fixed (but arbitrary) total order
(in Algorithm~\ref{alg:triangle-join} the attribute order is $A_0 \prec B_0 \prec C_0 \prec A_1 \prec B_1 \prec C_1$).
The attribute order defines how secondary index structures for each relation should be built for use in the join processing.
A side effect of the decision on attribute order is that excess indexes are sometimes required to be created.
Additionally, while the attribute order is immaterial to the worst-case analysis, it can
have a large effect on per-instance running time of a query. Similar to the use of backtracking in SAT solving
contexts, it turns out that the subqueries of the form above often
over-approximate candidate solutions. For example, in Algorithm~\ref{alg:triangle-join}, it is possible that the
size of a set $L_i$ is larger than the output size of the query on the given instance (but not larger than
the maximum output size of the query since it is worst-case optimal). To date,
the effect of attribute order on secondary index creation and per-instance runtime has received relatively 
little attention in the literature.

We extend gRTJ, by choosing the attribute order appropriately and by dealing with several bits in one step, so that we will be able to use
a single precomputed index per relation, i.e., an \emph{index-organised table}, and with this we can answer any join query on the given database instance.
Indexes can be precomputed without knowing the queries ahead of time. Avoiding the linear
runtime cost of index pre-computation particular to a query is already desirable since
if indexes are already in place some queries can be run in time even \emph{sub-linear} in the size of the input.
The specialised index-minimising algorithm has a runtime
overhead in the number of variables, yet remains worst-case optimal in the data complexity.  
This paper includes material which is in more detail in the second author's student thesis~\cite{brodythesis19}.

The contribution of this work is as follows:
\begin{itemize}
    \item a generic Radix Triejoin (gRTJ) algorithm that is non-com\-par\-i\-son based and data-structure independent  (in~Sec.~\ref{sec:radix-triejoin}),
    \item correctness of a subquery recurrence for reducing the search space (in~Sec.~\ref{chap:algorithm}),
    \item a new runtime analysis for an instance bound (in~Sec.~\ref{sec:analysis}),
    \item extending gRTJ for conjunctive queries with inequality constraints, and query-independent representations using bit-interleaving (in~Sec.~\ref{chap:extensions}).
\end{itemize}

\section{Preliminaries and Notation} \label{chap:preliminaries}
For the exposition of this paper, we follow the notation similar to~\cite{abiteboul1995foundations}. 
We assume the existence of a finite set of constants $\bU$ called the \emph{universe} or \emph{domain of discourse}.
A \emph{relation name} is a symbol $R$ associated with a finite attribute set
$\cA = \{A_1, \dots, A_r\}$, and is denoted by $R(\cA)$ or $R(A_1, \dots, A_r)$.
For a relation name $R(\cA)$, an \emph{$\cA$-tuple} (or simply \emph{tuple}) is a function $t : \cA \rightarrow \bU$.
Let $\bU^{\cA}$ denote the set of all $\cA$-tuples.
An \emph{$\cA$-relation} (or simply \emph{relation}) is a relation name $R(\cA)$ associated
with a subset of $\bU^\cA$. 
For a relation name $R(A_1, \dots, A_r)$ with an ordered attribute list,
it is common to identify a relation with a subset of the $r$-th Cartesian product $\bU^r$.
A \emph{schema} is a finite set of relation names $\cR$. 
A \emph{database} $\bD$ of schema $\cR$, or an \emph{$\cR$-database}, consists of an $\cA$-relation $R^\bD$ 
for every relation name $R(\cA)$ in $\cR$. For an attribute set $\mathcal{S}$, we write $t_\mathcal{S}$ to
denote the restriction of an $\cA$-tuple $t$ to $\mathcal{S}$. The \emph{projection} of a relation $R$
onto $\mathcal{S}$ is defined as
$
\pi_{\mathcal{S}} (R) = \{ t_{\mathcal{S}} \mid t \in R \}.
$
The \emph{semijoin} operator for two relations $R(\cA)$ and $S(\cB)$
is defined as
\[R \lJoin S := \{t \in R \mid \exists s \in S \text{ such that } \pi_{\cB}(t) = \pi_{\cA}(s)\}\]
that is, $R \lJoin S$ ``filters'' $R$ to only those tuples $t$ where there is a $\cB$-tuple $s$ in $S$ that joins
with $t$ on the common attributes.
A \emph{natural join
query} $Q$ (or simply \emph{query})
is specified by a schema $\{R_i(\cA_i)\}_{1 \leq i \leq m}$, and is written in the form
$Q := R_1 \Join \dots \Join R_m$ or $Q := \bigJoin_{1 \leq i \leq m} R_i$.
We write $\cA_Q := \cA_1 \cup \dots \cup \cA_m$ for the set of attributes of $Q$.
For a schema $\cR$ where $\{R_i(\cA_i)\}_{1 \leq i \leq m} \subseteq \cR$,
the \emph{answer} to a join query $Q$ on an $\cR$-database $\bD$
is denoted by $Q^\bD$,
and is defined as the set of exactly those $\cA$-tuples
$t$ whose projection onto the attribute set of each relation is an element of the relation. That is,
\begin{equation*} \label{eq:join-result}
Q^\bD := \left\{t \in \bU^{\cA_Q} \,\middle|\, R_i^\bD \lJoin t \neq \emptyset \ \text{for all}\  1 \leq i \leq m\right\}\text{.}
\end{equation*}
Potentially the result of $Q^\bD$ is as large as $|\bU|^{|\cA_Q|}$.
Often in this work we drop the database instance $\bD$ from $R_i^\bD$ and $Q^\bD$,
writing simply $R_i$ or $Q$ if the database is clear from context or implicit.
Henceforth we use also the convention that
$m$ is the number of relations and $n$
is the number of attributes of a query.

For the constants in the universe $\bU$, we introduce an encoding function 
$E : \bU \rightarrow {\mathbb B}^w$, which maps constants to 
bitstrings of length $w$. For a unique encoding, we 
require that $w \geq \lceil \log_2 |\bU| \rceil$. 
Given a database $\bD$
of domain $\bU$
and a $w$-bit encoding $E$ for $\bU$,
we can create a database $\bD^{(w)}$, which
is semantically equivalent with $\bD$.
We call $\bD^{(w)}$ the \emph{Booleanisation} of $\bD$
as it transforms $\bD$ with universe $\bU$ to a 
database in the Boolean universe $\{0, 1\}$.
Let $\cR = \{R_1(\cA_1), \dots, R_d(\cA_d)\}$
be a schema, and define $\cA_\cR := \cA_1 \cup \dots \cup \cA_d$ the set of attributes of the schema.
Let $w$ be an encoding length.
For each $A \in \cA_\cR$, we assume
the existence of a set $A^{(w)} := \{A_0, \dots, A_{w-1}\}$ of $w$
new attributes indexed from $0$ to $w-1$. Moreover,
for all $A \in \cA_\cR$
and $B \in \cA_\cR \setminus \{A\}$,
the sets $A^{(w)}$ and $B^{(w)}$ are required to be disjoint.
We define $\cA_i^{(w)} := \bigcup_{A \in \cA_i} A^{(w)}$.
For each $R_i$, let $R_i^{(w)}(\cA_i^{(w)})$
be a relation name on attribute set $\cA_i^{(w)}$ of a distinct relation symbol $R_i^{(w)}$.
The schema $\cR^{(w)} := \{R_1^{(w)}(\cA_1^{(w)}), \dots, R_d^{(w)}(\cA_d^{(w)})\}$
is called the \emph{$w$-th Booleanisation} of $\cR$.
For each $R_i$ of arity $r_i$ of the original schema $\cR$,
there
is a relation name $R_i^{(w)}$ of arity $r_iw$ in the Booleanisation $\cR^{(w)}$.
Given a join query $Q := \bigJoin_{1 \leq i \leq m} R_i$,
we also define the Booleanisation of a join query as $Q^{(w)} := \bigJoin_{1 \leq i \leq m} R_i^{(w)}$. 
Note that the Booleanisation of a database, unlike that of a schema, is not uniquely determined in general, as it depends on the encoding funciton $E$. The encoding function $E$ is closely related to the notion of an embedding between two relational structures \citep{gradel2007finite}.

\section{Natural Joins with generic Radix Triejoin} \label{chap:algorithm}

We describe the \emph{generic radix triejoin} (gRTJ) algorithm that uses Booleanisation to solve a natural join query $Q$ over a database $\bD = \{R_1, \dots, R_d\}$ of $d$ relations. The algorithm can be extended for complete conjunctive queries, see Section~\ref{chap:extensions}.
We prove correctness, and worst-case runtime optimality for gRTJ using an alternative proof strategy compared with the state-of-the-art.

We break gRTJ into two steps: (1) Booleanisation of a query, and (2) a backtracking/\emph{attribute-based} algorithm for solving the Booleanised query.
The Booleanisation transforms a query over a universe $\bU$ 
to an equivalent query over the Boolean universe ${\mathbb B}=\{0, 1\}$. The backtracking algorithm (Algorithm~\ref{alg:recursive-join}) solves the Booleanised
query by searching the attribute space. Initially, the algorithm starts with an empty candidate solution for its attributes, and in a given step fixes
one or more of the remaining attributes to concrete (Boolean) values by enumerating all combinations of their truth assignments.
Since the search space of attributes is reduced to single bits,
Booleanisation permits alternative search strategies for gRTJ compared with existing approaches.
We show that gRTJ solves
Booleanised queries worst-case optimally in the data complexity of the original query,
and has an
essentially equivalent dependency on the query expression-size terms to existing
worst-case optimal algorithms. Note that in the runtime expression of the other algorithms,
the encoding length $w$ is replaced by purely a $\log N$ term.

\begin{theorem} \label{thm:rtj-worst-case}
For a query $Q$ of $m$ relations, $n$ attributes
and relations of size $O(N)$ over a universe $\bU$,
gRTJ exhibits a worst-case runtime complexity of $O(mnw \cdot N^{\rho^*(Q)})$
where $w$ is the encoding length with $w = \lceil \log_2 |\bU| \rceil$ as the tightest encoding length.
\end{theorem}

We use \emph{hypergraphs} to model join queries that encode structural properties of the query (see ~\cite{ngo2013skew}).
For a join query $Q$ 
we construct a hypergraph $\cH = (\cV, \cE)$ where $\cV = \cA_Q$ is the set of
attributes of the query and there is a hyperedge $F \in \cE$
for each relation $R$ on an attribute set $F$. Structural
properties of queries such as cyclicity or more generally treewidth
are defined in terms of query hypergraphs.
We now denote a join query directly as a multi-hypergraph $\cH = (\cV, \cE)$.
A query is denoted as a formula $Q := \bigJoin_{F \in \cE} R_F$
where for each hyperedge $F \in \cE$, there is a distinct relation $R_F$
on attribute set $F$. Note that the hyperedges $F \in \cE$ are not necessarily
distinct if two different relations share the same attribute set.

The definition of
certain forms of \emph{subqueries} is central to deriving gRTJ.
In the framework of subqueries, Booleanisation unifies the theory of gRTJ with existing worst-case optimal algorithms as well.
These classes of
subqueries have first been introduced in
the introduction of the \emph{generic-join} framework \cite{ngo2013skew}.
At a high level, gRTJ (alongside LFTJ and NPRR) can be seen as a specialisation
of the generic-join framework. These algorithms have a simple recursive structure that can be considered as a divide-and-conquer approach.
Queries are divided into subqueries (subproblems) that are
recursively solved, and the solutions of the subqueries are combined to solve the original query.

In the following, we derive the generic-join framework formally
using a slightly different definition of subqueries to that
in \cite{ngo2013skew}.
The central result of the generic-join framework is a recurrence between solutions of the original query
and solutions of the subqueries.
We then specialise the generic-join framework to a specific solving strategy that gRTJ (and our later modified algorithm RTJ) uses.
The subqueries are defined based on partitioning the attributes of the query into two disjoint sets.
\begin{definition}[subqueries] \label{defn:subqueries}
Let $Q := \bigJoin_{F \in \cE} R_F$ be a query 
and $I \subseteq \cV$ be a subset of attributes of the query. Define the \emph{subqueries} of $Q$ as
\begin{align*}
Q_I &:= \bigJoin_{F \in \cE} \pi_I(R_F)\text{,} & \\
Q[t_I] &:= \bigJoin_{F \in \cE} \pi_{\cV \setminus I}(R_F \lJoin t_I) & \text{for all } t_I \in Q_I\text{.}
\end{align*}
\end{definition}

Subqueries are inspired by the \emph{splitting rule} of Boolean satisfiability solvers such as DPLL \citep{davis1961machine}.
The splitting rule only defines two smaller subproblems, one for each
truth value of a selected variable.
In the relational (i.e., predicate logic) case, there are possibly more than two subqueries--besides $Q_I$ itself,
there are $|Q_I|$ subqueries of the form $Q[t_I]$, one for each each solution $t_I$ of $Q_I$.
Nonetheless, using Booleanisation, if the attribute set $I$
is a singleton then there are only at most two solutions of $Q_I$. Moreover, the subqueries $Q[t_I]$
are induced by assigning a single attribute to a truth value, equivalent to the splitting rule.
We demonstrate the definitions of the subqueries in an example.
\begin{example}
Let $Q^{(2)}$ $:=$ $R^{(2)}(A_0, A_1, B_0, B_1)$ $\Join$ $S^{(2)}(B_0, B_1, C_0, C_1)$ $\Join$ $T^{(2)}(A_0, A_1, C_0, C_1)$ be the $2$-nd Booleanisation of the triangle query.
For $I := \{A_0\}$, the subqueries of $Q^{(2)}$ are
\begin{align*}
Q_{A_0}^{(2)} &= \pi_{A_0}(R^{(2)}) \Join \pi_\emptyset(S^{(2)}) \Join \pi_{A_0}(T^{(2)})\text{,} & \\
Q^{(2)}[A_0 \mapsto a_0] &=
\pi_{A_1B_0B_1} (\sigma_{A_0 = a_0}(R^{(2)})) 
\Join
S^{(2)} \Join \\
 & \quad \quad\quad\quad \pi_{A_1C_0C_1} (\sigma_{A_0 = a_0}(T^{(2)}))
, \label{eq:triangle-plan}
\end{align*}
where $a_0 \in \{0, 1\}$ is a truth value.
\end{example}
The subqueries are individually well-defined queries that have hypergraphs.
The hypergraphs of the subqueries can be expressed as subhypergraphs of the original query.
\begin{definition}
Let $\cH = (\cV, \cE)$ be a hypergraph.
For $I \subseteq \cV$, define $\cE_I := \{F \cap I \mid F \in \cE\}$.
The \emph{subhypergraph} induced by $I$
is $\cH_I := (I, \cE_I)$ with vertex set $I$
and edge set $\cE_I$.
\end{definition}

The following proposition 
shows that for an attribute subset $I$ the hypergraphs of the subqueries reduce to one of two cases.
\begin{proposition} \label{prop:subhypergraph}
Let $Q := \bigJoin_{F \in \cE} R_F$ be a join query
and $I \subseteq \cV$ be a subset of the attributes of the query.
The hypergraph of subquery $Q_I$ is $\cH_I$ and, for all $t_I \in Q_I$,
the hypergraph of subquery $Q[t_I]$ is $\cH_{\cV \setminus I}$ (i.e., dependent only on $I$ and not a specific tuple $t_I$).
\end{proposition}
\begin{proof}
By definition of projection, if $R_F$ is an input relation of $Q$,
then $\pi_I(R_F)$ is a relation of $Q_I$ on
attribute set $F \cap I$.
The set of attributes of $Q_I$ is then
$\bigcup_{F \in \cE} (F \cap I) = (\bigcup_{F \in \cE} F) \cap I = \cV \cap I = I$.
For each $t_I \in Q_I$,
since semijoin does not affect a query's hypergraph,
similar reasoning holds for the subquery $Q[t_I]$.
\end{proof}

In the following, we derive in Theorem~\ref{thm:recursive-form} a recurrence that
shows a relationship between solutions of the original query $Q$ and solutions of the
subqueries $Q_I$ and, for all $t_I \in Q_I$, $Q[t_I]$.
For all solutions $t_I$ of $Q_I$, the tuple $t_I$
can be combined with solutions of $Q[t_I]$
to form solutions of the original query $Q$.
We say that $Q_I$ controls a \emph{prefix space} of the query
and $Q[t_I]$ controls a \emph{suffix space} of the query for a particular $t_I \in Q_I$.
The Theorem~\ref{thm:recursive-form}
recurrence is the basis of a recursive divide-and-conquer algorithm of the next section.
We break the proof into two propositions.
First, Lemma~\ref{lem:query-decomp} states that the solutions of $Q_I$ over-approximates
the prefix space of $Q$. That is,
there can exist solutions $t_I$ of $Q_I$
which cannot be completed to a solution of $Q$.
\begin{proposition} \label{lem:overapprox}
Let $Q := \bigJoin_{F \in \cE} R_F$ be a query and $I \subseteq \cV$ be a subset of the attributes. Then,
\[\pi_I(Q) \subseteq Q_I.\]
\end{proposition}
\begin{proof}
First, note by Definition~\ref{defn:subqueries}, we expand $Q_I := \bigJoin_{F \in \cE} \pi_I(R_F)$. If $t_I \in \pi_I(Q)$,
then immediately by definition $t_I \in \pi_I(R_F)$ for all $F \in \cE$ so $t_I \in Q_I$.
\end{proof}

Second, Lemma~\ref{lem:resultsemi} shows a relationship between semijoin
and the subquery $Q[t_I]$ for a tuple $t_I \in Q_I$. 
\begin{proposition} \label{lem:resultsemi}
Let $Q := \bigJoin_{F \in \cE} R_F$ be a query
and $I \subseteq \cV$ be a subset of the attributes.
If $t_I \in Q_I$ then \[Q \lJoin t_I = \{t_I\} \times Q[t_I].\]
\end{proposition}
\begin{proof}
Let $t_I \in Q_I$.
First, observe that semijoin distributes over a natural join, i.e.,
\begin{equation}
Q \lJoin t_I = \bigJoin_{F \in \cE} (R_F \lJoin t_I)\text{.} \label{eq:resultsemiequiv}
\end{equation}
If $Q \lJoin t_I = \emptyset$ then clearly
$Q[t_I] = \emptyset$ so
$Q \lJoin t_I = \emptyset = \{t_I\} \times Q[t_I]$ as required.
If $Q \lJoin t_I \neq \emptyset$ then
$\pi_I(Q \lJoin t_I) = \{t_I\}$ by definition of semijoin since $I \subseteq \cV$.
Then,
\begin{align*}
\pi_{\cV \setminus I}(Q \lJoin t_I) &= \pi_{\cV \setminus I}\left( \bigJoin_{F \in \cE} (R_F \lJoin t_I) \right) \\
&= \bigJoin_{F \in \cE} \pi_{\cV \setminus I}(R_F \lJoin t_I) = Q[t_I]
\end{align*}
where the second equality is because all attributes in $I$ are bound by the semijoin.
Hence, $\pi_I(Q \lJoin t_I) = \{t_I\}$ and $\pi_{\cV \setminus I}(Q \lJoin t_I) = Q[t_I]$.
Thus, we have $Q \lJoin t_I = \{t_I\} \times Q[t_I]$.
\end{proof}

We prove the main recurrence that constructs the scaffolding for the divide-and-conquer (generic-join) algorithm.
We also apply the recurrence to derive an exact instance bound on the running
time of gRTJ. As far as the authors are aware, this is the first such result. 
\begin{lemma}[\cite{ngo2013skew}] \label{thm:recursive-form}
Let $Q$ be a query and let
$I \subseteq \cV$ be a subset of the attributes. Then,
\begin{equation}
Q = \bigcup_{t_I \in Q_I} (\{t_I\} \times Q[t_I])\text{.} \label{eq:recursive-form}
\end{equation}
\end{lemma}
\begin{proof}
By Lemma~\ref{lem:query-decomp}, $\pi_I(Q) \subseteq Q_I$.
Thus
\begin{align*}
Q &= \bigcup_{t_I \in Q_I} (Q \lJoin t_I) & \text{(over-approximation)}  \\ 
&= \bigcup_{t_I \in Q_I} (\{t_I\} \times Q[t_I]) & \text{(by Lemma~\ref{lem:resultsemi}).}
\end{align*}
\end{proof}

\subsection{gRTJ Algorithm} \label{sec:radix-triejoin}
We apply the recurrence as the basis of a backtracking
algorithm.
Importantly,
the solution space of the subqueries is reduced from the original query (provided
that we choose $I \subsetneq \cV$ and to be nonempty).
The divide-conquer-combine principles of the algorithm are outlined below:

\begin{itemize}[leftmargin=*]
\item[] \textbf{Divide:} Let $I \subseteq \cV$ be a \emph{singleton}. Solve $Q_I$ directly
by enumerating the possible Boolean solutions, and divide the query $Q$ into
subqueries (subproblems) $Q[t_I]$ for all $t_I \in Q_I$.
\item[] \textbf{Conquer:} Recursively solve the subqueries $Q[t_I]$ for all $t_I \in Q_I$.
\item[] \textbf{Combine:}
Apply the recurrence to combine the solutions to subqueries.
A combined prefix $t_I \in Q_I$ and suffix $t_{\cV \setminus I} \in Q[t_I]$ forms one solution of the query.
We take the union of all prefixes with all of the respective suffixes.
\end{itemize}
The recursion also terminates with the base case when $|\cV| \leq 1$ for which 
we solve the query of at most one attribute by enumeration of its possible Boolean solutions.
Overall, the divide-and-conquer framework
reduces the full query to mere Boolean cases that
can be solved directly. Note that the sets $Q_I$
guide the search for full solutions.
We present in Algorithm~\ref{alg:recursive-join} the full algorithm.
The gRTJ algorithm has a sub-procedure \textsc{BasicSolve} (Algorithm~\ref{alg:basicsolve}) that solves subqueries of the form $Q_I$ by full enumeration of its possible solutions.
Note that full enumeration (a.k.a. guessing, brute-force)
is only practical due to the reduction of the size of the universe
to two by Booleanisation, i.e., zero and one. 

\begin{algorithm}[t]
\caption{Subprocedure: \textsc{BasicSolve}($Q := \bigJoin_{F \in \cE} R_F$)}
\begin{algorithmic}[1]
\Require{Query $Q$, hypergraph $\cH=(\cV,\cE)$}
\State \Return{$\{t \in \{0, 1\}^{\cV} \mid R_F \lJoin t \neq \emptyset$ for all $F \in \cE$\}} 
\end{algorithmic}
\label{alg:basicsolve}
\end{algorithm}

\begin{algorithm}[t]
\caption{High-level view: \textsc{gRTJ}($Q := \bigJoin_{F \in \cE} R_F$)}
\begin{algorithmic}[1]
\Require{Query $Q$, attribute order $\cV = \{A_1, \dots, A_n\}$}
\If {$|\cV| \leq 1$} \label{alg:recursive-join:base-condition}
  \State \Return{\Call{BasicSolve}{$Q$}}
\EndIf
\State Let $I = \{A_1\}$ \label{alg:recursive-join:isubset}
\State $L \gets \Call{BasicSolve}{Q_I}$ \Comment{solve $Q_I$ directly} \label{step:generic-last-call}
\State \Return{$\bigcup_{t_I \in L} (\{t_I\} \times \Call{gRTJ}{Q[t_I]})$} \label{alg:recursive-join:recursion} 
\end{algorithmic}
\label{alg:recursive-join}
\end{algorithm}

In full generality, Algorithm~\ref{alg:recursive-join} conceivably
picks any attribute subset $I$ to define the subqueries at a given recursion.
As a result, the algorithm is not entirely specified and more like a \emph{family} of possible algorithms.
In practice,
a strong technical assumption of all current worst-case optimal algorithms
(one which we follow in the remainder of this work)
is that the attributes are processed in a fixed
total order such as $A_1 \prec \dots \prec A_n$.
First,
choosing the subsets $I$ of the attributes in a deterministic manner according to the recursion level
simplifies the analysis.
Second, and more importantly, 
the attribute order determines the access characteristics of the indexes associated to each relation.
For example, data structures of relation indexes for gRTJ are bitwise tries,
OBDDs, or quadtrees for some examples. However, each of these data structures must be created
with a certain attribute order.
To illustrate the effect of attribute order,
in Figure~\ref{fig:left-rtree} and \ref{fig:right-rtree},
we show running the algorithm in two possible attribute orders on the triangle query
and identical Booleanised database $R^{(1)}(A_0, B_0) = \{(0, 0), (1, 0), (1, 1)\}$,
$S^{(1)}(B_0, C_0) = \{(0, 1), (1, 0)\}$,
$T^{(1)}(A_0, C_0) = \{(0, 1), (1, 0)\}$. The nodes in the recursion tree
represent invocations of the sub-routine \textsc{gRTJ}. Each
recursion level in the recursion tree expands one of the attributes. 
The attribute order $(A_0, C_0, B_0)$
results in the smaller recursion tree in terms of the number of nodes. 
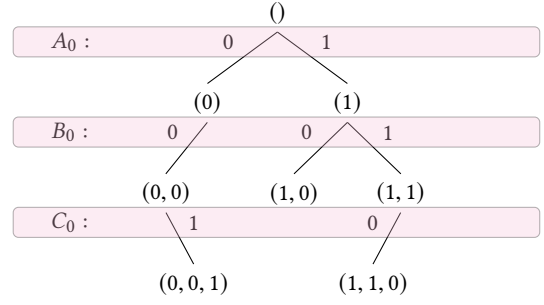
\begin{figure}
\centering
\begin{tikzpicture}[level distance=1.2cm]
\Tree
[.{$()$}
	\edge node[auto=right] (level0-right) {$0$};
	[.{$(0)$}
		\edge node[auto=right] (level1-right) {$0$};
		[.{$(0, 0)$}
			\edge[draw=none]; \node[draw=none] {};
			\edge node[auto=left] (level2-right) {$1$};
			[.{$(0, 0, 1)$} ]
		]
		\edge[draw=none]; \node[draw=none] {};
	]
	\edge node[auto=left] {$1$};
	[.{$(1)$}
		\edge node[auto=right] {$0$};
		[.{$(1, 0)$} ]
		\edge node[auto=left] {$1$};
		[.{$(1, 1)$}
			\edge node[auto=right] {$0$};
			[.{$(1, 1, 0)$} ]
			\edge[draw=none]; \node[draw=none] {};
		]
	]
]
\foreach \Value/\Text in {0/{$A_0:$},1/{$B_0:$},2/{$C_0:$}}
{  
  \node[anchor=west] 
    at ([xshift=-20mm]{level2-right}|-{level\Value-right}) 
    {\Text};
   \draw[fill=magenta!30,opacity=0.3,rounded corners=2] ([xshift=-24mm,yshift=-2mm]{level2-right}|-{level\Value-right})  rectangle ++(7,0.4);
}
\end{tikzpicture}
\vspace{-3mm}
\caption{Recursion tree for attribute order $(A_0, B_0, C_0)$.}
\label{fig:left-rtree}
\end{figure}
~
\begin{figure}
\centering
\begin{tikzpicture}[level distance=1.2cm]
\Tree
[.{$()$}
	\edge node[auto=right] (level0-right) {$0$};
	[.{$(0)$}
		\edge[draw=none]; \node[draw=none] {};
		\edge node[auto=left] (level1-right) {$1$};
		[.{$(0, 1)$}
			\edge node[auto=right] (level2-right) {$0$};
			[.{$(0, 1, 0)$} ]
			\edge[draw=none]; \node[draw=none] {};
		]
	]
	\edge node[auto=left] {$1$};
	[.{$(1)$}
		\edge node[auto=right] {$0$};
		[.{$(1, 0)$}
			\edge[draw=none]; \node[draw=none] {};
			\edge node[auto=left] {$1$};
			[.{$(1, 0, 1)$} ]
		]
		\edge[draw=none]; \node[draw=none] {};
	]
]
\foreach \Value/\Text in {0/{$A_0:$},1/{$C_0:$},2/{$B_0:$}}
{  
  \node[anchor=west] 
    at ([xshift=-20mm]{level2-right}|-{level\Value-right}) 
    {\Text};
   \draw[fill=magenta!30,opacity=0.3,rounded corners=2] ([xshift=-24mm,yshift=-2mm]{level2-right}|-{level\Value-right})  rectangle ++(7,0.4);
}
\end{tikzpicture}
\vspace{-3mm}
\caption{Recursion tree for attribute order $(A_0, C_0, B_0)$.
} \label{fig:right-rtree}
\end{figure}
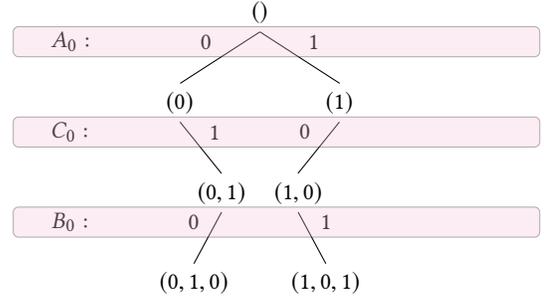

We show the correctness (for an arbitrary attribute order) of the gRTJ algorithm via the recurrence (Theorem~\ref{thm:recursive-form}).
\begin{theorem} \label{thm:rtjoin-correctness}
For all input databases and all queries $Q$,
the output of gRTJ($Q$) (Algorithm~\ref{alg:recursive-join}) is equal to the answer to $Q$.
\end{theorem}
\begin{proof}
By induction on the number of attributes $n \geq 1$.
Note that \textsc{BasicSolve} is simply the definition of the join result,
and hence the base case is correct.
Let $n \geq 2$. The induction hypothesis is that the algorithm computes
queries of $n - 1$ attributes correctly.
In particular, for a query $Q$ of $n$ attributes, the subquery $Q[t_I]$ is computed
correctly for all $t_I \in Q_I$.
It follows from from Theorem~\ref{thm:recursive-form} that the algorithm computes the output of $Q$ correctly
as well. By induction, the correctness of the algorithm holds.
\end{proof}

\subsection{Runtime Analysis} \label{sec:analysis}

We derive an instance bound that exactly characterises the per-instance runtime complexity of gRTJ,
then apply the AGM bound to the instance bound to prove worst-case optimality of gRTJ.
Note also that the proof techniques in this section more generally apply to algorithms that
have the same recursive structure as gRTJ (e.g., LFTJ), and at a high level, we need not necessarily assume a Boolean universe.

\begin{figure*}
\centering
\begin{tikzpicture}[level distance=1.5cm]
\Tree
[.{$Q_{A_1}$}
	\edge node[auto=right] (level0-right) {$0$};
	[.{$(Q[A_1 \mapsto 0])_{A_2}$}
		\edge node[auto=right] (level1-right) {$0$};
		[.{$(Q[A_1 \mapsto 0, A_2 \mapsto 0])_{A_3}$}
			\edge node[auto=right] (level2-right) {$0$};
			[.{$\vdots$} ]
			\edge node[auto=left] {$1$};
			[.{$\vdots$} ]
		]
		\edge node[auto=left] {$1$};
		[.{$(Q[A_1 \mapsto 0, A_2 \mapsto 1])_{A_3}$}
			\edge node[auto=right] {$0$};
			[.{$\vdots$} ]
			\edge node[auto=left] {$1$};
			[.{$\vdots$} ]
		]
	]
	\edge node[auto=left] {$1$};
	[.{$(Q[A_1 \mapsto 1])_{A_2}$}
		\edge node[auto=right] {$0$};
		[.{$(Q[A_1 \mapsto 1, A_2 \mapsto 0])_{A_3}$}
			\edge node[auto=right] {$0$};
			[.{$\vdots$} ]
			\edge node[auto=left] {$1$};
			[.{$\vdots$} ]
		]
		\edge node[auto=left] {$1$};
		[.{$(Q[A_1 \mapsto 1, A_2 \mapsto 1])_{A_3}$}
			\edge node[auto=right] {$0$};
			[.{$\vdots$} ]
			\edge node[auto=left] {$1$};
			[.{$\vdots$} ]
		]
	]
]
\foreach \Value/\Text in {0/{$A_1:$},1/{$A_2:$},2/{$A_3:$}}
{  
  \node[anchor=west] 
    at ([xshift=-25mm]{level2-right}|-{level\Value-right}) 
    {\Text};
   \draw[fill=magenta!30,opacity=0.3,rounded corners=2] ([xshift=-29mm,yshift=-2mm]{level2-right}|-{level\Value-right})  rectangle ++(17,0.4);
}
\end{tikzpicture}
\vspace{-3mm}
\caption{An unspecific recursion tree that shows the amount of work performed at each call
for the attribute order $A_1, \dots, A_n$.}
\label{fig:recursion-work-tree}
\end{figure*}
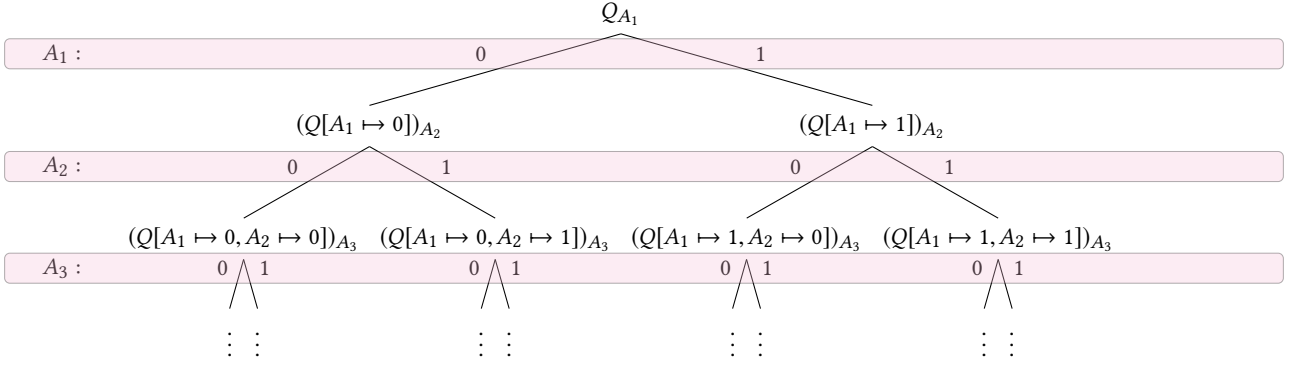

It is worthwhile to note that the
attribute order is immaterial to worst-case performance --
it only affects per-instance performance.
Similarly, choice of encoding for the Booleanisation affects only per-instance performance.
To derive an instance bound,
we follow a \emph{recursion-tree method} (see \cite{cormen2009introduction}).
We sum over the amount of work performed at all levels of the recursion tree
induced by the recursive calls of gRTJ.
We visualise a generic recursion tree in Figure~\ref{fig:recursion-work-tree}
using an attribute order $A_1, \dots, A_n$.
In Figure~\ref{fig:recursion-work-tree}, recursive calls of the recursion tree
are labelled by the subquery that is computed at the call.
Note that the recursion tree is a binary tree due to Booleanisation.
Define $I_j := \{A_1, \dots, A_j\}$.

The following sub-lemma is a simple application of Theorem~\ref{thm:recursive-form} to subqueries.
\begin{lemma} \label{lem:query-decomp}
Let $Q$ be a query and
$I, J \subseteq \cV$ be disjoint subsets of the attributes of the query.
Then,
\begin{equation}
Q_{I \cup J} = \bigcup_{t_I \in Q_I} (\{t_I\}\ \times (Q[t_I])_J)\text{.} \label{eq:query-decomp}
\end{equation}
\end{lemma}
\begin{proof}
We use the fact that that $(Q_{I \cup J})[t_I] = (Q[t_I])_J$.
The result follows from replacing $Q$ by $Q_{I \cup J}$ in Theorem~\ref{thm:recursive-form}. 
\end{proof}

Note that the amount of work performed at a single recursive call at a recursion level $j$, not including the costs
of subsequent recursive calls, is the time it takes to compute a subquery of the form
$(Q[t_{I_j}])_{A_{j + 1}}$ and is represented by the labels of nodes in the recursion tree.
Our key assumption is that the appropriate relation indexes
are available so that a
subquery $(Q[t_{I_j}])_{A_{j + 1}}$ can be computed in time $O(m \cdot (Q[t_{I_j}])_{A_{j + 1}})$ in the \textsc{BasicSolve} procedure, i.e., $O(1)$ time to query each available index.
The second sub-lemma characterises the total amount of work performed at a recursion level $j$.
\begin{lemma} \label{thm:recursion-tree}
The total amount of work performed by gRTJ (Algorithm~\ref{alg:recursive-join}) at level $j$ in the recursion tree is $O(m \cdot |Q_{I_{j+1}}|)$.
\end{lemma}
\begin{proof}
First, by the assumption the total amount of work done at level $j$ is represented
by the size of the union of the results to subqueries computed at level $j$,
which we proceed to show by induction is equal to $|Q_{I_{j+1}}|$.
First, the initial call at level $0$ corresponds to the original query $Q$ 
and computing the subquery $Q_{I_1}$.
For the induction hypothesis, we assume
there is only a recursive call on $Q[t_{I_j}]$ at level $j$ for each of the solutions
$t_{I_j}$ of the subquery $Q_{I_j}$.
Then, for each $t_{I_j} \in Q_{I_j}$, there are children of the recursive call $Q[t_{I_j}]$
at level $j + 1$ with call parameter $Q[t_{I_j} \cup t_{A_j}]$
for each $t_{A_j} \in (Q[t_{I_j}])_{A_{j+1}}$.
Thus,
we have in total that there is a respective recursive call at level $j + 1$
for each solution of $\bigcup_{t_{I_j} \in Q_{I_j}}
(\{t_{I_j}\} \times
(Q[t_{I_j}])_{A_{j + 1}}) = Q_{I_{j+1}}$, where here we have applied Lemma~\ref{lem:query-decomp}.
By induction, the total amount of work performed at level $j$ is
$|\bigcup_{t_{I_j} \in Q_{I_j}}
(Q[t_{I_j}])_{A_{j + 1}}| = |Q_{I_{j+1}}|$, where here we have again applied Lemma~\ref{lem:query-decomp}.
\end{proof}

We set up an \emph{instance bound} for the running time of the algorithm
that characterises exactly the running time of the algorithm on a particular query and database instance.
We express the per-instance running time in terms of the sizes of the subqueries $Q_{I_j}$ of $Q$.
The following result is relevant to \emph{beyond worst-case guarantees}.
Also observe that $I_j = \{A_1, \dots, A_j\}$ in the theorem is dependent on the attribute order,
so the running time also is dependent on the attribute order.
\begin{theorem} \label{thm:rtj-exact-t}
Let $Q$ be a join query of $m$ relations and $n$ attributes and $Q^{(w)}$ be
the Booleanisation of the query for an encoding length $w$.
Then gRTJ (Algorithm~\ref{alg:recursive-join}) runs in time
$
O \left( m \cdot \sum_{j = 1}^{nw} |Q_{I_j}| \right) \text{.}
$
\end{theorem}
\begin{proof}
We observe that the height of the recursion tree is $nw - 1$.
For each attribute in the original query $Q$,
there are $w$ new distinct attributes in $Q^{(w)}$.
We sum
over the amount of work performed at each level $0 \leq j \leq nw - 1$
of the recursion tree to get the total running time of order of
$\sum_{j=0}^{nw-1}(m \cdot |Q_{I_{j+1}}|)
= m \cdot \sum_{j=1}^{nw} |Q_{I_{j}}|.$
\end{proof}

To get a worst-case bound, we use the AGM bound to bound the per-instance running time.
Recall that the AGM bound is the maximum output size of a join query $Q$
in terms of the relation sizes.
In particular, if all relations have the same size $N$, the AGM bound is
$|Q| \leq N^{\rho^*(Q)}$.
For the subqueries $Q_{I_j}$, we have $|Q_{I_j}| \leq N^{\rho^*(Q_{I_j})}$.
To prove worst-case optimality, i.e.,
a running time of $O(mnw \cdot N^{\rho^*(Q)})$, we first require a lemma
bounding the fractional edge cover number $\rho^*(Q_{I_j})$
of a subquery $Q_{I_j}$ by the fractional edge cover number $\rho^*(Q)$ of the original query $Q$.
It has first been noted by \cite{grohe2006constraint}
when discussing join-project plans.
\begin{lemma}[\citeauthor*{grohe2006constraint}] \label{thm:wco-lemma}
Let $Q := \bigJoin_{F \in \cE} R_F$ be a query and let $I \subseteq \cV$
be a subset of the attributes of the query.
Then $Q_I$ has a fractional edge cover number bounded by $Q$. That is,
$
\rho^*(Q_I) \leq \rho^*(Q)\text{.}
$
\end{lemma}
\begin{proof}
From Proposition~\ref{prop:subhypergraph}, the hypergraph of the subquery $Q_I$ is
$\cH_I := (I, \cE_I)$
where $\cE_I = \{F \cap I \mid F \in \cE\}$.
We show that a fractional edge cover of the hypergraph $\cH$ of $Q$
gives a fractional edge cover of $\cH_I$.
Then, since the fractional edge cover {number} $\rho^*(Q_I)$
is the minimum over all fractional edge covers,
we have $\rho^*(Q_I) \leq \rho^*(Q)$.
Given a fractional edge cover for $\cH$,
for each hyperedge $F \in \cE$,
we set the weight of $F \cap I \in \cE_I$ to the weight of $F$.
The two covers have the same total cost.
The cover of $\cH_I$ is valid as well since $\cH_I$ only removes attributes
from the original hypergraph $\cH$.
\end{proof}

We now bound the worst-case running time of gRTJ by the AGM bound of a query.
The running time is worst-case optimal in terms of data complexity.
\begin{proof}[Proof of Theorem~\ref{thm:rtj-worst-case}]
We have, by applying the AGM bound and Theorem~\ref{thm:wco-lemma},
$
|Q_{I_j}| \leq N^{\rho^*(Q_{I_j})} \leq N^{\rho^*(Q)}\text{.}
$
Hence by Theorem~\ref{thm:rtj-exact-t},
the running time of the algorithm is bounded by order of
$
m \cdot \sum_{j = 1}^{nw} |Q_{I_j}| \leq mnw \cdot N^{\rho^*(Q)}\text{.}
$
\end{proof}

The running time
of the algorithm is dependent on the input query and database, unsurprisingly.
It is also dependent on the encoding function of elements in the universe
and the attribute expansion order.
While the query and database are a fixed input, encoding and attribute order represent
two possible degrees of freedom that can be chosen by the algorithm or an optimiser.
In terms of the worst-case bound,
encoding and attribute order can be arbitrary.
For example, the \emph{set intersection} query expressed as
$
Q(A) := R(A) \Join S(A)
$
returns all values in the intersection of $R$ and $S$.
If $R$ and $S$ are in fact disjoint, then
for certain choices of encoding and attribute order
gRTJ can determine in $O(1)$ time that the result is empty.

\section{gRTJ Extensions} \label{chap:extensions}

We consider two extensions of the gRTJ algorithm:
First, the gRTJ algorithm is extended to handle
a more general class of queries formulated in first-order logic
known as \emph{full conjunctive queries}. Full conjunctive queries
express the most frequently occurring queries in databases in practice \citep{chandra1977optimal}.
We have also worst-case optimality extended to full conjunctive queries. Full details of this approach can be found in~\cite{brodythesis19}.

Second, we remove the requirement for the generic framework
to create query-dependent secondary indexes to answer queries.
We introduce a \emph{query-independent} relation representation such that
no additional secondary indexes are required to answer queries.
We modify gRTJ to (1) expand more than one bit at each recursive call,
and (2) restrict the possible attribute orders to a subclass we call the bitwise \emph{interleaved orders}.
We also apply a strongly connected components algorithm to produce an optimal expansion order
given a fixed attribute order (i.e., index) for each relation representation.
The resulting algorithm which we call \emph{Radix Triejoin} (or RTJ) incurs a query-dependent runtime overhead.
For a full conjunctive query $Q$ of $m$ atoms and $n$ variables over relations of size $N$,
the running time of RTJ is
$
O \left( 2^n m w \cdot N^{\rho^*(Q)} \right)\text{.}
$

\subsection{Conjunctive Queries with Inequality.} \label{sec:conjunctive}

The class of \emph{conjunctive queries}
generalise the natural join queries of relational algebra
as expressions in first-order logic. Conjunctive queries express
satisfying assignments for variables rather than attributes.
While a join query is defined over a set of relations
each on a fixed attribute set,
a conjunctive query is defined over a set of atomic formulas.
An \emph{atomic formula} (or simply \emph{atom}) in our context is a formula $R(u)$ where $R(\cA)$
is an $r$-ary relation
and $u$ is a tuple with domain $\cA$ of (not necessarily distinct) variables or constants.
A conjunctive query consists of
a set of atoms.
Conjunctive queries have the following generalisations over join queries:
\begin{enumerate}
\item \emph{Repeated Variables:} The same variable can occur more than once in the same atom;
\item \emph{Repeated Relations:} The same relation can appear as the predicate of different atoms in one query;
\item \emph{Constants:} The arguments of relations can be constants as well as variables.
\end{enumerate}

\citeauthor*{gottlob2012size} extended the ideas of AGM
and derived upper and lower bounds for the sizes of conjunctive query results
and also considered functional dependencies \cite{gottlob2012size}.
In short, we construct a hypergraph $\cH = (\cV, \cE)$
for a conjunctive query where $\cV$ is the set of variables of the query
and there is a hyperedge $F \in \cE$ for every atom in the rule with variable set $F$.
The fractional edge cover bound of AGM also
derives a tight bound on the output size of full conjunctive queries.
If $\rho^*(Q)$ is the fractional edge cover number of $\cH$,
then $|Q| \leq N^{\rho^*(Q)}$ is a tight output size bound for relations of size $N$.

It is relatively straightforward to adapt gRTJ to the three generalisations above
of full conjunctive queries. For example, repeated relations can be handled
purely syntactically by a rewrite of the query formula to remove duplicates. Constants constrain the search tree
rather than the case of a multi-valued variable. Similarly, the second or more occurrences
of a variable in an atom are constrained by a concrete variable binding at an earlier level
in the search tree.

\paragraph*{\Large{Partial Match Queries.}}
A particularly interesting relationship of RTJ is with
partial match retrieval algorithms
for which the literature is huge \citep{rivest1976partial}.
\citeauthor*{flajolet1986partial} (FP henceforth) studied the average-case complexity
for partial match retrieval in binary kd-tries \cite{flajolet1986partial}.
Partial match retrieval corresponds to a
conjunctive query on a single atom over a Boolean universe, i.e.,
full conjunctive queries of the form
\begin{equation}
Q(u_0) \leftarrow R(u_1)\text{,} \label{eq:pmq}
\end{equation}
where $R$ is a relation of arity $r$ and $u_1$ is a tuple of distinct variables
and (not necessarily distinct) constants.
In the partial match retrieval literature, the elements of $u_1$ that are constants (variables)
are called the \emph{specified} (\emph{unspecified}) components.
Interestingly, when the variable expansion order for RTJ is an interleaved order,
the FP average-case analysis applies to RTJ as well.
An interleaved order expands variables of a query cyclically in the Booleanisation.
For example, if $x_1, \dots, x_n$ are variables of the original query,
then the algorithm expands the least significant bit of $x_1$ first,
followed by the least significant bit of $x_2$, \ldots, and wraps
back around to $x_1$ to continue to expand the second least significant bit, \ldots, etc.
The implied constant hidden in the big-$O$ notation of the following theorem is small.
\begin{theorem}[\citeauthor*{flajolet1986partial}]
The average cost, measured by the number of internal nodes of the corresponding recursion tree, of a
partial match query of $s$ specified constants
constructed from a relation of size $N$ and arity $r$ under the Bernoulli model
is $O(N^{1 - s / r})$.
\end{theorem}

\paragraph*{\Large{Inequality Constraints}}
A conjunctive query can be extended with inequality constraints for the variables.
We use properties of the encoding and Booleanisation to adapt RTJ to conjunctive queries
with inequalities. Inequality constraints constrain the search tree, possibly
improving the per-instance runtime of RTJ.
For example, suppose an inequality constraint is $x \leq 4$
in the original query,
and integer values of the universe are encoded using $w$ bits
in the binary representation.
The algorithm deduces that higher-order bits
for satisfying assignments of $x$ are equal to $0$, otherwise the constraint
will be falsified. Suppose variable $x$
is turned into $w$ new variables $x_0, \dots, x_{w-1}$
in the Booleanisation. We deduce based on the encoding
that the inequality constraint $x \leq 4$,
in the Booleanisation, implies the constraints $x_3 = 0,
\dots, x_{w-1} = 0$.
Hence inequality constraints constrain the recursion tree
and allow gRTJ to eliminate ranges of the search space, i.e.,
subtrees of the recursion tree, from having to be explored.

\subsection{Query-Independent Representations.} \label{sec:query-independence}

We adapt gRTJ to execute conjunctive queries using a \emph{query-in\-de\-pen\-dent} relation representation.
Note that in the Section~\ref{chap:algorithm} version of gRTJ and existing worst-case optimal algorithms,
an \emph{order-consistent} secondary index is assumed for each relation.
For a fixed variable order,
order consistency means that for each atom $R(u)$ of a query,
there is a secondary index for $R$ built in an ordering on the attributes of $R$
that is compatible with (or induced by) the variable order.
The induced order is the total order on the attributes $\cA$ of relation $R$
such that, for all $A, B \in \cA$, if variable $u(A)$ precedes variable $u(B)$ in the variable order,
then $A$ precedes $B$ in the induced attribute order.
However, to achieve order consistency requires creating potentially a new secondary
index for every occurrence of a relation in queries.
Note that there are $r!$ unique secondary indexes for a relation of arity $r$.
In this section we adapt gRTJ to remove the requirement to create multiple secondary indexes
and instead we store relations only in \emph{index-organised tables} \citep{shkapsky2016big}.
With a query-independent relation representation, we achieve the possibility of
\emph{sub-linear} queries and memory efficiency as well. The index-organised tables are
precomputed only once during database setup, rather than in query evaluation.

\paragraph{Bitwise Interleaved Orders.}

To achieve a query-independent relation representation, we weaken the order-consistency
assumption and modify gRTJ to expand more than one bit at a time.
For a Booleanisation $Q^{(w)}$ in the original formulation of gRTJ, any of the
$(nw)!$ permutations of the variable set are valid variable orders.
Now we restrict the variable orders to a subclass called the \emph{interleaved}
variable orders.
Assuming relations are also stored in an interleaved attribute order,
we have
an upper bound of $n$ bits that need to be expanded
at each recursive call in order to utilise the index of each relation.
The additional runtime cost incurred at each recursive call is $2^n$ to expand all truth
assignments of $n$ bits.

In the following, we assume an already Booleanised query of
$nw$ variables and write $Q$ instead of $Q^{(w)}$ and
for relations $R$ instead of $R^{(w)}$. The variable set
of $Q$ is $\cV = \{x_i^{(j)} \mid 1 \leq i \leq n, \, 0 \leq j \leq w-1\}$, where
$x_1, \dots, x_n$ are variables of the original (non-Booleanised) query.
We assume also for simplicity that the full conjunctive query does not
include constants or repeated variables in atoms.
The attribute set of a relation $R$ is
$\cA = \{A_i^{(j)} \mid 1 \leq i \leq r, \, 0 \leq j \leq w-1\}$,
where $A_1, \dots, A_r$ are attributes of the original (non-Booleanised) relation.
An \emph{interleaved} variable order is obtained by merging the variable sequences
$(x_i^{(j)})_{0 \leq j \leq w - 1}$ for $1 \leq i \leq n$ into a single sequence via a shuffle.
Formally, a permutation $\alpha$ of $\{1, 2, \dots, n\}$
defines a distinct interleaved order on the variables of $Q$.
The interleaved order given $\alpha$ on the variable set $\cV$ of query $Q$
is represented by a variable sequence
\begin{equation}
\begin{array}{lll}
\ell_\cV &=& x_{\alpha(1)}^{(0)} \prec \dots \prec x_{\alpha(n)}^{(0)} \prec x_{\alpha(1)}^{(1)}
\prec \dots \prec x_{\alpha(n)}^{(1)} \prec \dots \prec \\
  & & x_{\alpha(1)}^{(w-1)} \prec \dots \prec x_{\alpha(n)}^{(w-1)}\text{,} \label{eq:intervarorder}
\end{array}
\end{equation}
where the variables of $Q$ corresponding to the least significant bits
of encodings of elements come first, followed by the second least significant bits, and so on, cyclically.
There are $n!$ unique interleaved variable orders in total.
Similarly, the relation representations are constructed
in an interleaved order on the attributes of the relation. For a relation $R$ of
arity $rw$ in the Booleanisation, there are $r!$ interleaved orders of the attribute set
and have two fundamental observations.
We define the prefix set of a sequence.
\begin{definition}[prefix set]
Given a sequence
$\ell = y_1 \prec y_2 \dots \prec y_l$ and integer $1 \leq k \leq l$,
the \emph{$k$-th prefix} of $\ell$ is
$\{y_1, \dots, y_k\}$.
\end{definition}

\begin{lemma} \label{obs:prefixset}
The $n$-th prefix of an interleaved variable order $\ell_\cV$
is $\{x_1^{(0)}, \dots, x_n^{(0)}\}$,
and the $r$-th prefix of an interleaved attribute order
for a relation of arity $rw$ is $\{A_1^{(0)}, \dots, A_r^{(0)}\}$.
\end{lemma}

\begin{lemma} \label{obs:interleaved}
Let $R(u)$ be an atom of a Booleanised conjunctive query $Q$, where
relation $R$ has attribute set $\cA$.
For the $n$-th prefix $I := \{x_1^{(0)}, \dots, x_n^{(0)}\}$ of an interleaved variable order $\ell_\cV$,
the preimage $u^{-1}[I] = \{A \in \cA \mid u(A) \in I\}$
is the $r$-th prefix $\{A_1^{(0)}, \dots, A_r^{(0)}\}$ of any interleaved attribute order (i.e., any index) $\ell_\cA$.
\end{lemma}

\paragraph{Radix Triejoin Algorithm.} 

Lemma~\ref{obs:prefixset} and \ref{obs:interleaved} imply a modification to gRTJ to work with
the query-independent relation representation.
We use a bitwise interleaved attribute order. In each recursive call,
we expand $n$ variables at a time,
i.e., add $n$ bits to the candidate solution. The algorithm with these modifications is called \emph{Radix TrieJoin} (RTJ).
By the lemmas, a variable binding for the $n$-th prefix of the variable order
induces an attribute binding on a prefix set of the attribute order of each relation
so that the corresponding index is queryable. Thus, using bitwise interleaved orders,
we achieve the desired query-independent relation representation.
We also note the following changes to the worst-case analysis. As each
recursive call expands now $n$ variables instead of only one variable,
the height of the recursion tree is $nw / n - 1 = w - 1$. 
Since there are $2^n$ truth assignments for $n$ Boolean variables,
the amount of time in a recursive call is
increased to $O(2^n m)$ where $m$ is the number of atoms.
In total, the worst-case runtime of the RTJ algorithm is
\[
O \left( 2^n m w \cdot N^{\rho^*(Q)} \right).
\]

\begin{example}
Suppose there is a relation $R(A, B)$
and $2$ bits is sufficient to encode the database universe.
Let an index for the Booleanisation $R^{(2)}(\cA)$ be built in an interleaved attribute
order of the attribute set $\cA = \{A_0, A_1, B_0, B_1\}$ such as
$\ell_{\cA} = A_0 \prec B_0 \prec A_1 \prec B_1$.
For a conjunctive query
\[
Q(x, y) \leftarrow R(x, y) \wedge R(y, x)\text{,}
\]
we have a $2$-nd Booleanisation
\[
Q^{(2)}(x_0, x_1, y_0, y_1) \leftarrow R^{(2)}(x_0, x_1, y_0, y_1) \wedge R^{(2)}(y_0, y_1, x_0, x_1)\text{.}
\]
RTJ answers the query using a single attribute order (i.e., index) $\ell_{\cA}$.
For example,
in the first recursive call RTJ finds truth values for variables $x_0$ and $y_0$, i.e.,
it expands two bits. In the atom $R(x_0, x_1, y_0, y_1)$,
we have $A_0 \mapsto x_0$ and $B_0 \mapsto y_0$
and so query $(x_0, y_0)$ in $\ell_{\cA}$.
On the other hand, in the atom $R(y_0, y_1, x_0, x_1)$,
we have $A_0 \mapsto y_0$ and $B_0 \mapsto x_0$
and so query $(y_0, x_0)$ in the same index.
\end{example}

\paragraph{Optimal Variable Expansion Orders.}
Up to now the choice of interleaved variable
order (the parameter $\alpha$) of the $n!$ possibilities is irrelevant as we always
expand the $n$-th prefix at each
recursive call.
We now consider a more sophisticated attempt
that takes into account how the indexes of each relation are interleaved with respect to the attributes.
We expand the minimum set of variables at each recursive
call such that the indexes are still queryable.
In this model the attribute orders of relations are fixed upfront.
We develop an algorithm to produce
an optimal variable expansion order given a query and the fixed attribute orders.

The attribute orders impose order constraints on the possible variable expansion orders.
Suppose an index on relation $R$ with attribute set $\cA$ is $\ell_\cA = A_1 \prec \dots \prec A_r$.
Then for all atoms $R(u)$ in the query on relation $R$, for all $1 \leq i \leq r - 1$,
we require that variable $u(A_i)$ precedes $u(A_{i+1})$ in the variable expansion order
(written $u(A_i) \prec u(A_{i+1})$).
It is possible that the constraints imposed on $\ell_\cV$ conflict, i.e.,
for variables $x, y$, one atom imposes
the constraint $x \prec y$ and another atom imposes the constraint $y \prec x$.
To accommodate the order conflict
we expand both variables $x$ and $y$ at the same time,
using the idea that an index can be queried so long as a tuple is defined on a
prefix set of the index.
In the remainder of this section
we develop an optimal \emph{variable expansion order} that designates
an ordering on a partition of the variable set $\cV$ into disjoint sets $S_1, \dots, S_k$, whose
meaning is that in a $j$-th recursive call, RTJ expands the variable set $S_{j+1}$.

Intuitively, suppose there is a conflict in the induced order on $\ell_\cV$
such as $x \prec y$ and $y \prec x$.
As mentioned, we expand variables $x$ and $y$
at the same time. We add them to a set $S = \{x, y\}$. Moreover, for all variables $z \neq y$ such that
$x \prec z$ and $z \prec x$, the algorithm expands all of $S \cup \{z\}$ at the same time as well.
In general,
we are required to find a partition of $\cV$ into $k$ sets $S_1, \dots, S_k$
such that
for all $x, y \in \cV$, if $x \prec y$ and $y \prec x$, then $x, y$ are elements of the same set $S_j$.

The problem is reducible to the strongly connected components (SCCs) of a graph. We create a
directed graph $G = (V, E)$ where
the vertex set $V$ is the set of variables $\cV$ and
for each constraint $x \prec y$ of $\ell_\cV$,
there is an edge $(x, y)$ in the graph.
To find a partition $\mathcal{S} = \{S_1, \dots, S_k\}$ of
$\cV$ such that the previous condition holds
is now equivalent to the SCCs of $G$.
We run an algorithm to compute the SCCs
such as Tarjan's \citep{tarjan1972depth}, which runs in time $O(|V| + |E|)$.
As there are $m$ atoms of the query
and each atom involves at most $n = |\cV|$ variables,
there are at most $m(n-1)$ edges in $G$, one for each pairwise order constraint $u(A_i) \prec u(A_{i+1})$. We
thus compute the set of SCCs $\cS$ in time $O(|V| + |E|) = O(n + mn) = O(mn)$ time.
Moreover, the topological order on the SCCs
is a valid variable expansion order for RTJ.
The topological order is
$\ell_\cS = S_1 \prec \dots \prec S_k$ on $\cS$
such that, if $x \prec y$ for variables $x \in S_i$ and $y \in S_j$,
then $S_i \prec S_j$. Note that Tarjan's
algorithm also topologically sorts the SCCs as a side effect \citep{knuth1993stanford}.
Thus, we compute in $O(mn)$ time the optimal
variable expansion order $\ell_\cS$.

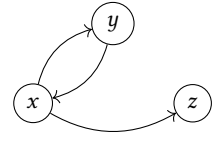
\begin{figure}
\begin{tikzpicture}[every node/.style={circle,draw},node distance=1.5cm]
\node (x) at (0,0) {$x$};
\node [above right of=x] (y) {$y$};
\node [below right of=y] (z) {$z$};

    \path [->, bend left] (x) edge (y);
    \path [->, bend left] (y) edge (x);
    \path [->, bend right] (x) edge (z);
\end{tikzpicture}
\vspace{-3mm}
\caption{Induced graph for the indexes in
Example~\ref{ex:fullconjqind}.} \label{fig:sccs}
\end{figure}

\begin{example} \label{ex:fullconjqind}
We consider the full conjunctive query
\[
Q(x, y, z) \leftarrow R(x, y) \wedge S(y, x, z)\text{,}
\]
where $R(A, B)$ and $S(A, B, C)$ are relations.
We assume the query is already a Booleanisation,
i.e., the encoding length is $1$ for simplicity.
Suppose there are fixed indexes on the relations
$\ell_R = R.A \prec R.B$
and $\ell_S = S.A \prec S.B \prec S.C$.
We represent the variable equivalences
as equivalence classes on the attributes of the relations.
As there are three variables, we have three
equivalence classes enforcing equalities between attributes of the relations:
$x$: $\{R.A, S.B\}$,
$y$: $\{R.B, S.A\}$, and
$z$: $\{S.C\}$.
The indexes $\ell_R$ and $\ell_S$
lift to order constraints on the equivalence classes.
We have $x \prec y$, $y \prec x$, and $x \prec z$
from the index constraints $R.A \prec R.B$,
$S.A \prec S.B$, and $S.B \prec S.C$ respectively.
We form a graph $G = (V, E)$ with
$V = \{x, y, z\}$ and edge
set $E = \{(x, y), (y, x), (x, z)\}$.
The graph is displayed in Figure~\ref{fig:sccs}.
The topological order we compute on the SCCs of $G$ is
$\{x, y\} \prec \{z\}$. Thus,
the optimal variable expansion order (given the fixed indexes)
corresponds to expanding two variables $x$ and $y$ at recursion level $0$.
At recursion level $1$, we expand the singleton variable $z$.
\end{example}

Let a query $Q$ be a Booleanisation of an encoding length $w$,
over $nw$ variables and $m$ atoms.
Let $\ell_\cS = S_1 \prec \dots \prec S_k$
be the optimal variable expansion order where $\cS = \{S_1, \dots, S_k\}$ is
a partition of variable set $\cV$ for query $Q$ into disjoint nonempty subsets.
The amount of time spent in a recursive call at level $j$ is $O(2^{|S_{j+1}|} m)$
to expand the full set of truth assignments of $|S_{j+1}|$ variables.
By Theorem~\ref{thm:rtj-worst-case}, the runtime of RTJ for the variable expansion order $\cS$
and relations of size $N$ is
\begin{equation}
O \left(
m \sum_{j = 1}^{k} 2^{|S_j|} \cdot N^{\rho^*(Q)} \label{eq:specificexp}
\right)\text{.}
\end{equation}
To minimise the above, we want to minimise the size of the subsets $|S_j|$.
Equivalently, since the $S_j$ are a partition of $\cV$, we maximise $k$.
In particular, if we use interleaved variable orders,
we can derive an upper bound.
We observe that for the interleaved order
$|S_j| \leq n$ by Lemma~\ref{obs:interleaved} for all $1 \leq j \leq k$.
Let $p := \max_{1 \leq j \leq k} |S_j|$.
Since $|S_j| \leq p$ for all $1 \leq j \leq k$
and $|S_1| + \dots + |S_k| = nw$, we have
an upper bound for Equation~\ref{eq:specificexp} as
$
O(m 2^p (nw / p) \cdot N^{\rho^*(Q)})
$
where $1 \leq p \leq n$
is a query- and index-dependent factor.
Hence, we achieve a better runtime bound using the SCCs approach
rather than always expanding $n$ variables.
Note that if
$|S_j| = n$ for all $1 \leq j \leq k$,
we retrieve the previous bound, and
if $|S_j| = 1$ for all $1 \leq j \leq k$ (so $p = 1$),
we retrieve the original worst-case runtime of Theorem~\ref{thm:rtj-worst-case}.

\section{Discussion}

We will compare various algorithms related to RTJ.
For example, the Leapfrog Triejoin (LFTJ) algorithm has been developed and deployed within commercial applications before the theoretical lower bound for worst-case optimal join algorithms have been discovered in 2012~\cite{veldhuizen2012leapfrog}. Retrospective analysis revealed the algorithm's worst-case optimality. The Leapfrog Triejoin algorithm is a generalization of leapfrog joins for sorted lists to relations of higher arity. In the list case, with is equivalent to unary relations, a set of $n$ relations is joined by repeatedly progressing iterators over the lists. For instance, to join the three lists  $l_1 = [1,5,7]$, $l_2 = [2,4,5,8]$, and $l_3 = [1,3,5,7]$ the algorithm starts by obtaining an iterator to the first (smallest) element of each list. Let $i_1$ to $i_3$ represent those iterators and e.g. $i_1 \rightarrow 1$ denote the state where the first iterator references the element $1$ in $l_1$. Thus, the algorithm starts with state $(l_1,l_2,l_3) \rightarrow (1,2,1)$. In the next step, iterators are sorted according to the value they reference, resulting in $(l_1,l_3,l_2) \rightarrow (1,1,2)$. Also the maximal referenced element $m=2$ is computed. After the initial phase, the iterator referencing the smallest element is updated to point to the smallest element within its respective list that is not smaller than the current maximum element $m$. Thus, in our example the following sequence of states is processed:
\begin{align*}
    (l_1,l_3,l_2) &\rightarrow (1,1,2) & m=2\\
    (l_1,l_3,l_2) &\rightarrow (5,1,2) & m=5 \\
    (l_1,l_3,l_2) &\rightarrow (5,5,2) & m=5 \\
    (l_1,l_3,l_2) &\rightarrow (5,5,5) & m=5 
\end{align*}
Note that in each step the iterator referencing the smallest element is always the successor of the last updated iterator. In cases where all iterators reach a common element, an element of the join result is obtained. In the sequence above, $5$ is yielded as a result. Furthermore, any of the iterators is moved to the respective successor element. Once, any iterator reaches the end of the list, denoted by $\bot$, the algorithm terminates. Thus, after processing
\begin{align*}
    (l_1,l_3,l_2) &\rightarrow (5,5,7) & m=7\\
    (l_1,l_3,l_2) &\rightarrow (7,5,7) & m=7 \\
    (l_1,l_3,l_2) &\rightarrow (7,8,7) & m=8 \\
    (l_1,l_3,l_2) &\rightarrow (7,8,\bot) & m=8 
\end{align*}
the algorithm terminates.

The leapfrog join is generalised to the Leapfrog Triejoin algorithm to support relations with more than one attribute.
For a given query, e.g. $Q(x,y,z) \leftarrow R(x,y) \wedge S(y,z), \wedge T(x,y)$, a variable order is fixed, e.g. $(x,y,z)$, and a backtracking based enumeration of all satisfying variable assignments conducted. Thus, in the given example, a leapfrog join is conducted for values of $x$ such that $R(x,\_) \wedge T(x,\_)$ is valued. For each value $x=x^\prime$ identified by the join operation, a recursive join enumerating values for the variable $y$ under the constraint that $x=x^\prime$ is initiated. If such a $y=y^\prime$ is found, a third enumeration for values $z=z^\prime$ under the constraint $(x,y)=(x^\prime,y^\prime)$ is conducted. For each value $(x^\prime,y^\prime,z^\prime)$ is yielded as an element of the query result.

The algorithmic efficiency of the algorithm builds on the ability to perform efficient lower-bound queries benefiting from iteratively reduced search spaces due to the gradual introduction of value constraints in each recursive step of the algorithm. To ensure this, input relations are required to be stored in tries (or equivalent index structures) according to the variable order determined for the execution of the query -- in the example above $(x,y,z)$. The required index order on various relations is thus query dependent, and queries like $Q(x,y) \leftarrow R(x,y) \wedge R(y,x)$ require multiple indexes on the same relation.

To the contrast, our RTJ supports efficient processing of arbitrary queries using a single index on each relation due to the option of interleaving the binary encoding of attributes. Nevertheless, both algorithms' per-instance efficiency depends on the variable ordering. In the Leapfrog case the ordering of the actual query variables, in the RTJ case on the ordering of the bits in the encoding. In both cases, the difference between a good and a bad ordering can cause the difference between an instance-optimal and a worst-case optimal query execution. 

Also, both algorithms are based on the idea of enumerating suitable variable assignments. Leapfrog does so by systematically searching the input data, while RTJ is gradually building up the binary encoding of suitable values. Furthermore, although both algorithms refer to tries in their names, their definition is widely independent of the data structure utilized for maintaining the processed relation data -- as long as a set of run-time complexity constraints of certain operations on the relation structure are guaranteed.

A major difference between LFTJ and RTJ is the ability to support multiple occurrences of the same variable within single terms of queries. For instance, the query $Q(x) \leftarrow R(x,x)$ can be directly supported by RTJ, while LFTJ requires a reformulation into $Q(x) \leftarrow R(x,y) \wedge I(x,y)$ where $I$ is a non-materialized identity relation. The support for constants in formulas, like $Q(x) \leftarrow R(x,1)$ is realized by both algorithms through the introduction of a non-materialized relation $C = \{1\}$ and a rewrite into $Q(x) \leftarrow R(x,y) \wedge C(y)$.

In addition to the need of rewriting the input query to fit LFTJ restrictions on variable usage, the lack of supporting multiple occurrences of variables as arguments for the same relation can also negatively affect run-time efficiency. For instance, for the query $Q(x) \leftarrow R(x,x)$ and the database
$
R = \{ (2i, 2i+1) \mid 0 \leq i \leq n \}
$
where $n$ is a natural number, the result of the query is clearly empty. Applying RTJ with an interleaved least-significant-bit first bit-order is able to determine the emptiness of the result $Q$ in $O(1)$ steps. LFTJ, however, has to process the rewritten query $Q(x) \leftarrow R(x,y) \wedge I(x,y)$, where $I$ is a non-materialized identity relation. During evaluation it will bind the variable $x$ to each value in $\{ 2i \mid 0 \leq i \leq n \}$, which covers $n$ elements. For each of those, the LFTJ performs a recursive step to attempt to identify a corresponding value for $y$, which fails in $O(1)$ steps, leading to an overall complexity of $O(n)$ steps. Choosing the alternative variable order of $(y,x)$ leads to the same result. For the given query and database instance, RTJ is instance optimal, while LFTJ is worst-case optimal.

\subsection{Comparison with DPLL Algorithm}

The bit-wise backtracking based nature of our algorithm bears similarities with the DPLL algorithm forming the foundation for many SAT solving tools. Like our algorithm, DPLL is based on the recursive exploration of the assignment space of boolean variables to their binary domain $\left\{\text{true},\text{false}\right\}$ or $\left\{0,1\right\}$. Furthermore, full conjunctive queries can be converted into a (variant) of a SAT solving problem, facilitating the utilization of DPLL for solving those.

For instance, a 1-bit instance of the triangular query
$
Q^{(1)} := R^{(1)}(A_0, B_0) \Join S^{(1)}(B_0, C_0) \Join T^{(1)}(A_0, C_0)\text{,}
$
with relations
\begin{align*}
R(A_0, B_0) &= \{(0, 0), (0, 1), (1, 0)\}\text{,} \\
S(B_0, C_0) &= \{(0, 1), (1, 1)\}\text{,} \\
T(A_0, C_0) &= \{(0, 0), (1, 0)\}\text{.}
\end{align*}
can be encoded into a SAT problem by searching for a satisfying assignments of the boolean variables $A_0$,$B_0$, and $C_0$ of the constraints
\begin{align*}
&((\neg A_0 \wedge B_0) \vee (\neg A_0 \wedge B_0) \vee(A_0 \wedge \neg B_0)) \wedge \\
&((\neg B_0 \wedge C_0) \vee (B_0 \wedge C_0)) \wedge \\
&((\neg A_0 \wedge \neg C_0) \vee (A_0 \wedge \neg C_0)) 
\end{align*}
where each row encodes the content of one of the three relations. Every satisfying assignment of variables yields a different element of the result set $Q$. Since it is a mere enumeration of set entries, the length of this encoding is $O(N*|V|)$ where $N$ is the size of the input relations and $V$ the set of variables.

To solve the example using DPLL the given propositional formula where relations are encoded using Disjunctive normal form (DNF) needs to be converted into Conjunctive normal form (CNF) -- which is co-NP-hard. Alternatively, relations could directly be encoded in CNF form by enumerating the disjunction of the negation of missing elements, thus
\begin{align*}
&(\neg A_0 \vee \neg B_0) \wedge \\
&(B_0 \vee C_0) \wedge (\neg B_0 \vee C_0) \wedge \\
&(A_0 \vee \neg C_0) \wedge (\neg A_0 \vee \neg C_0)
\end{align*}
for the example. This is not circumventing the exponential complexity of converting DNF into CNF formulas -- since the resulting formula contains up to $2^a$ terms for each atom of the full conjunctive query, where $a$ is the number of variables in the corresponding atom. 

The outlined encoding enables the conversion of query joins into a format amendable to the DPLL algorithm. However, in its basic formulation SAT solving problems are merely concerned on determining a single satisfying assignment -- or proofing a lack thereof. In join terms, the algorithm merely determines whether the resulting set is empty or not, and if not, provides a single element as proof. DPLL is designed for determining satisfiability and is thus not comparable to our join algorithm.

However, a variant of SAT, known as All-SAT, asks for an exhaustive enumeration of all satisfying assignments. Applying a solver for this variant to the CNF formula derived above yields the desired result. DPLL can be customized to address this problem formulation. For the remainder of this section we assume a corresponding adaptation, to obtain comparability between DPLL and RTJ.

Both algorithms, DPLL and RTJ, are based on utilizing a back-tracking scheme for exploring the range of possible boolean assignments for query variables. However, DPLL exhibits three major advances over RTJ in this regard: the ability to select variables at arbitrary order, the ability to conduct back-jumps, and the ability of clause learning.

In RTJ the order in which variables are bound during recursive processing is fixed by the order used for constructing input relations. To the contrary, DPLL has the freedom of choosing which variable to bind in each step. In theory, this ability allows DPLL to always obtain ideal variable ordering. For instance, in our previous example, binding $C_0$ first would immediately imply that the given formula is unsatisfiable, thus the query result is empty. However, in order to harness this potential in practice, heuristics determining variable good variable orders are required -- often involving the computation of statistical data on the structure of the processed formula. 

The ability to flexibly chose which variable to bind is furthermore  utilized by state-of-the-art DPLL implementations to facilitate backjumping -- an advanced variant of backtracking. For instance, if RTJ is using the variable order $A_0$, $B_0$, $C_0$ it (might) start with $A_0 = 0$, followed by $B_0=0$ only to determine that neither $C_0 =0$ nor $C_0=1$ will produce a result. It would thus backtrack up to the decision of $B_0$ and continue with $B_0=1$ only to learn that there isn't any solution either. Advanced variants of DPLL, on the other hand, can determine that $A_0=0$ implies $C_0=0$ when detecting the first time that the $(A_0,B_0,C_0)=(0,0,0)$ is not a satisfying solution. Furthermore, if it detects that there is not value for $B_0$ such that $C_0=0$ can lead to a solution, it may conclude that the branch of $B_0=1$ can be skipped, immediately continuing with $A_0=1$. The ability to jump back multiple steps in the back-tracking algorithm can greatly reduce the search space to be covered.

Finally, the ability to learn is closely related to the reasoning performed by the back-jumping process outlined above. The conflict analysis conducted before performing a backjump yields additional information on implied constraints between variables. This information, in the form of additional constraints, is added to the processed formula to avoid exploring the search space later in a different branch constituting the same root problem causing the currently processed conflict. Thus, with learning, DPLL is extended by the ability to modify the processed formula, by adding terms forming logical consequences of the available terms.

The difficulty of DPLL, of cause, is the development of heuristics for deciding which variables to bind, the algorithms for performing conflict analysis, and the decision process on what learned clauses to retain or dismiss. DPLL is thus more like a family of algorithms, then a specific instance of it.

Nevertheless, while discussing similarities and differences between DPLL and RTJ is interesting from an algorithmic point of view -- potentially spawning new ideas for refining either of those -- naively applying DPLL for solving relational join problems is clearly not a viable approach in the general case. As hinted above, the encoding of relational data into CNF may cause and exponential blow-out in the amount of data -- and thus the necessary runtime to process it. Smarter encoding schemes exploiting the ability of propositional logic to describe implicit data sets yielding much more compact CNF representation could be utilized. Nevertheless, pathological cases cannot be avoided.

\subsection{Lazy qdag-based WCOJ Algorithm}
Recently, a data-structure driven worst-case optimal join algorithm design has been presented~\cite{navarro2019optimal}. The algorithm is based on a generalized region quadtree utilizing sharing for maintaining common sub-trees -- referred to as a \textit{qdag}. Tuples within $n$-ary relations are interpreted as $n$-dimensional points, and relations are stored within those qdags as sets of points similar to the way $n$-dimensional points would be stored in generalized Quadtrees. However, their data structure design facilitates the sharing of sub-trees, as well as a set of direct data structure manipulations corresponding in their effect to relational operations on the presented set of points. Among others, union, intersections, (restricted) cross products, and complements can be effectively computed. Furthermore, a combination of the supported operation yielding a worst-case optimal join algorithm is presented~\cite{navarro2019optimal}.

By interpreting relational data as points of an n-dimensional space and asserting integer coordinates, the presented algorithm corresponds to a geometrical interpretation of our booleanization -- assuming the natural encoding and bit-wise interleaving of attribute values. Applying the qdag based algorithm and our RTJ algorithm to correspondingly encoded data yields a sequence of practically identical processing steps. Additionally, in contrast to LFTJ, both algorithms only require a single index structure for each input relation to effectively support arbitrary queries on top of those.

The major difference between the qdag-based algorithm and RTJ is the structure considered as the main element to be manipulated. The qdag based algorithm considers relations as the basic structure on which operations need to be performed to obtain query results. Thus, relational operations like joins, unions, complements and cross products are executed. RTJ, like LFTJ and DPLL based solutions, focuses on query variables and their potential bindings. RTJ recursively builds up satisfying assignments for query variables, to ultimately reach a complete enumeration of those. RTJ's variable focused point of view enables processing of queries like
$Q(x) \leftarrow R(x,x)$ where the query variable $x$ shows up more than once within a single relation $R$. Recursively constructing values of $x$ in the given query is directly supported by our algorithm. However, no operation on the relational level would yield this result. Consequently, the relation-focused qdag-based algorithm requires some additional pre/post processing for these kinds of queries. While this processing step does not lead to worse asymptotic runtime complexity of queries, this observation demonstrates that RTJ's capabilities form a super-set of the qdag-based algorithm.

\section{Conclusion} \label{chap:conclusion}

We presented a practical ``worst-case optimal'' multi-way join algorithm
called the \emph{radix triejoin}.
The algorithm uses booleanisation and is not comparison-based.
It uses binary representation of values in a database
to perform domain reductions resulting in a simplified index structure.
Suitable interleaving of the boolean attributes
makes the relation representation \emph{query independent}:
the ability to compute multiple queries over the same relations 
using just one precomputed index per relation, while still having worst-case running time that matches the AGM bound up to some small factors (similar to prior algorithms).

\bibliographystyle{plainnat}
\bibliography{biblio}

\appendix
\addtocontents{toc}{\protect\setcounter{tocdepth}{1}}

\end{document}